\documentclass[pdflatex,sn-mathphys-num]{sn-jnl}

\usepackage{graphicx}%
\usepackage{multirow}%
\usepackage{amsmath,amssymb,amsfonts}%
\usepackage{amsthm}%
\usepackage{mathrsfs}%
\usepackage[title]{appendix}%
\usepackage{xcolor}%
\usepackage{textcomp}%
\usepackage{manyfoot}%
\usepackage{booktabs}%
\usepackage{algorithm}%
\usepackage{algorithmicx}%
\usepackage{algpseudocode}%
\usepackage{listings}%
\usepackage{lineno}

\usepackage{amsmath,amssymb,amsthm}

\newtheorem{theorem}{Theorem}[section]
\newtheorem{lemma}[theorem]{Lemma}
\newtheorem{proposition}[theorem]{Proposition}
\newtheorem{corollary}[theorem]{Corollary}

\theoremstyle{definition}
\newtheorem{definition}[theorem]{Definition}
\newtheorem{assumption}[theorem]{Assumption}

\theoremstyle{remark}
\newtheorem{remark}[theorem]{Remark}

\begin{document}

\title[Game-Theoretic Foundations of Competition for Conscious Access]{Game-Theoretic Foundations of Competition for Conscious Access}


\author*[1]{\fnm{Efthyvoulos} \sur{Drousiotis}}\email{Efthyvoulos.Drousiotis@liverpool.ac.uk}

\author*[2]{\fnm{Paul} \sur{Spirakis}}\email{P.Spirakis@liverpool.ac.uk}

\author*[3]{\fnm{Sotiris} \sur{Nikoletseas}}\email{nikole@cti.gr}

\affil*[1]{\orgdiv{Health Data Science Department}, \orgname{University of Liverpool},  \country{UK}}

\affil*[2]{\orgdiv{School of Computer Science and Informatics}, \orgname{University of Liverpool},  \country{UK}}

\affil*[3]{\orgdiv{Computer Engineering and Informatics Department}, \orgname{Patras University},  \country{Greece}}


\abstract{\begin{abstract}

 Conscious access in the human brain is often described as the outcome of a competition among candidate representations, but this competition is usually left at the level of
mechanism or metaphor rather than analyzed as a strategic allocation problem. We introduce an access contest in which internal modules compete for a scarce broadcast slot by choosing a costly amplification effort. Access is allocated by a
smooth probabilistic rule, allowing the model to interpolate between diffuse selection and winner-take-all competition.

We establish pure-strategy equilibrium existence under standard convexity and bounded-benefit assumptions, and give sufficient conditions for uniqueness using diagonal
strict concavity. We then analyze capture in the two-module case, where a lower-value representation can obtain a larger access probability because it is
easier to amplify. For quadratic costs, we derive a sharp threshold in the competition intensity above which capture occurs. More generally, for strongly
convex costs, we prove an if-and-only-if capture criterion in terms of the cost-adjusted amplification advantage of the lower-value module. Under the same curvature-dominance condition that guarantees uniqueness, we show that the unique pure Nash equilibrium of the general \(M\)-module access contest can be approximated efficiently by projected pseudo-gradient dynamics, with logarithmic dependence on the desired accuracy. Finally, we prove an impossibility theorem for single-slot access mechanisms. Exact winner-take-all efficiency is incompatible with robustness to small score perturbations. Thus, smooth probabilistic access rules are not merely analytically convenient, but structurally motivated.

These results provide a game-theoretic foundation for studying competition for conscious access, connecting equilibrium analysis, capture, computation, and mechanism-level limitations under a common formal model.
\end{abstract}}


\maketitle

\section{Introduction}

\subsection{Motivation}

The brain performs substantial computation outside conscious awareness. Stimuli can bias decisions, guide actions, and shape learning even when they are not available for report. Clinical and experimental dissociations, such as blindsight and related phenomena, show that information can be processed and used without becoming consciously accessible. This motivates a basic theoretical
distinction. Information may be represented and acted upon locally, while only some representations become globally available for flexible reasoning, reporting, and control.

A common way to describe this transition is through competition. Many candidate representations are processed in parallel, but conscious access is capacity-limited. Only a small number of representations can enter the global
workspace or broadcast channel at a given time. This raises a natural allocation problem. When several internal representations compete for access, what determines which one becomes conscious?
For example, consider a student trying to solve a mathematical problem while
simultaneously worrying about an emotionally charged social interaction. The
mathematical representation may have greater task value, yet the emotionally
salient representation may be easier to amplify and may therefore dominate
conscious access. The question is then whether access tracks value, or whether
it can be captured by lower-value but more easily amplified content.

Despite the centrality of this idea, the competition for access is often left at the level of metaphor. Existing theories describe amplification, ignition, selection, or global broadcast, but they do not usually specify the access mechanism as a strategic allocation rule. As a result, several basic questions remain difficult to answer formally. When should competition for access have a stable outcome? When is the outcome unique? When does access track the most valuable representation, and when can it be captured by a lower-value but more easily amplified representation? Are there structural limitations that any single-slot access rule must satisfy?
Global workspace accounts already motivate this abstraction. They describe many specialised unconscious processes operating in parallel and competing for access
to a limited workspace before selected information is globally broadcast
\cite{baars1988cognitive,shanahan2005applying}. We use game-theoretic language in this functional sense where modules are not assumed to deliberate consciously, but are modelled as competing amplification processes.

The aim is not to provide a complete theory of consciousness, but to isolate one formal component of conscious access, the competition for a scarce broadcast slot. By modelling this competition explicitly, we obtain equilibrium, capture, computational, and impossibility results that clarify when access behaves efficiently and when it can become distorted.

\subsection{Relevant Work}

A central distinction in consciousness research separates phenomenal
consciousness from access consciousness. Phenomenal consciousness concerns
subjective experience, whereas access consciousness concerns the availability
of information for report, reasoning, and behavioural control
\cite{block1995confusion}. Empirical dissociations such as blindsight show that
information can influence behaviour without being consciously reportable
\cite{weiskrantz1986blindsight}. Our focus is on this access component, the
mechanism by which some representations become globally available while others
remain locally processed.

Since the early development of psychoanalysis, several scientists have attempted to move beyond qualitative notions of the mind towards more formal and structural models of unconscious processes. Freud introduced an “economic” abstraction of the psyche in which mental processes are captured via tensions, quantities, and the circulation of psychic energy. In his "Project for a Scientific Psychology" (1895), he described the mind as a dynamic system attempting to regulate and reduce internal tension. The human psyche was treated as a system that allocates limited psychic resources across competing desires and representations. Lacan, rather than focusing on energy models, described the unconscious as “structured like a language,” emphasizing symbolic relations, recursive structures, and networks of signifiers, using graphs, topology and logical structures.

Bion developed a theory according to which raw emotional and sensory experiences (“beta elements”) initially cannot be mentally processed and must be transformed through the “alpha function” into thinkable mental content (“alpha elements”)~\cite{bion1970attention}. In this sense, the mind operates almost like a transformation system acting on unprocessed inputs under emotional and cognitive constraints. This gives Bion’s theory an almost computational structure: the mind transforms unprocessed inputs under conditions of limited containment and emotional regulation. Failures of transformation lead to fragmentation, intrusive experiences, or attacks on linking, whereas successful transformation allows symbolic thought, dreaming, abstraction, and reflective consciousness. In this sense, Bion anticipated later computational and information-processing approaches that model cognition as the dynamic transformation and integration of internal representations under constraints.

Although these approaches remained largely conceptual rather than computational, they anticipated later attempts to model cognition and unconscious dynamics using systems theory, information theory, dynamical systems, and computational frameworks

Global workspace theory and global neuronal workspace theory provide influential
accounts of conscious access. In these theories, specialised processors operate
in parallel, while conscious access corresponds to the selection, amplification,
and broadcast of information to a wider cognitive system
\cite{baars1988cognitive,dehaene2001cognitive,dehaene2011gnw,mashour2020gnw}.
These theories naturally use the language of competition, selection, ignition,
and broadcast. However, the competition is usually described mechanistically
rather than as an explicit strategic allocation problem with equilibrium
predictions, capture conditions, or mechanism-level guarantees.

A related literature in attention studies selection under limited processing
capacity. Biased-competition models describe attention as a competition among
representations, with bottom-up salience and top-down goals biasing selection
\cite{desimone1995attention}. Experimental work on attentional capture shows
that salient or contextually favoured stimuli can dominate selection even when
they are not task-optimal
\cite{folk1992contingent,theeuwes2010topdown}. This literature motivates the
study of capture, but it typically does not model competing representations as
strategic agents with effort costs and best responses.

Game-theoretic tools have also been used to formalise social and cognitive
phenomena. For example, the volunteer's dilemma has been used to explain the
bystander effect. When several witnesses observe an emergency, each individual
may prefer that someone intervene while hoping that another witness bears the
cost of doing so \cite{diekmann1985volunteers,camposmercade2021bystander}.
Game theory has also been used in neuroscience and cognition, especially in social decision-making, neuroeconomics, and theory-of-mind modelling
\cite{lee2008gametheory,yoshida2008gtom}. Other work uses evolutionary game dynamics to model large-scale brain activity \cite{madeo2017egtbrain}. These approaches show that game-theoretic formalisms can be useful beyond explicit markets or economic settings, but they do not primarily study the access bottleneck between unconscious processing and conscious availability.

Mathematically, our model is closest to contest theory and probabilistic choice.
Contest theory studies agents who expend costly effort to win a prize under a
contest success function
\cite{tullock1980efficient,skaperdas1996contest,corchon2007contests,vojnovic2016contest}. Logit
and softmax rules are central in probabilistic choice and quantal response
equilibrium \cite{mckelvey1995qre}. Rational inattention and salience models
also treat attention or information processing as scarce and costly
\cite{sims2003rationalinattention,bordalo2012saliencerisk,bordalo2013salienceconsumer}.
However, these models are usually not applied to internal competition among
unconscious processors for a shared access channel.

A further connection is to the algorithmic game theory literature, where
equilibrium behaviour, computational tractability, and inefficiency under
strategic behaviour are central themes. In our setting, the access mechanism
induces a normal-form game in which each module's effort affects both its own
access probability and the access probabilities of the others. This places the
model close to classical questions about existence and uniqueness of equilibria,
as well as to the broader algorithmic question of whether equilibria can be
computed efficiently in structured games. Our analysis exploits this structure by giving sufficient conditions under which the
unique pure Nash equilibrium of the general access contest can be approximated
efficiently by projected pseudo-gradient dynamics.

The model also has a natural connection to asymmetric contests. In standard
contest models, asymmetries in valuations or costs can substantially affect the
allocation of effort and the probability of winning. This is important for the
present paper because capture is precisely an asymmetric phenomenon. A
lower-value representation may nevertheless dominate access if it has a
sufficient cost advantage in amplification. Contest theory provides the general
language of costly effort and probabilistic success
\cite{tullock1980efficient,skaperdas1996contest,corchon2007contests}, but the
interpretation here is different. The competing agents are not external
decision-makers competing for an external prize, but internal modules competing
for access to a shared conscious broadcast channel.

Despite these connections, several aspects of competition for conscious access
remain underdeveloped. Existing access theories rarely specify an allocation
rule whose equilibrium behaviour can be analysed. Attention and salience models
study capture, but not usually as an equilibrium consequence of costly strategic
amplification. Contest and attention-economics models provide the right
mathematical tools, but they are not typically applied to internal competition
among unconscious processors for a shared access channel. Finally, there is
limited formal understanding of what properties any single-slot access
mechanism can satisfy simultaneously, such as efficiency, robustness,
monotonicity, and differentiability.

\subsection{Contributions}

These gaps motivate the present paper. We develop a contest-theoretic model of the access bottleneck and analyse its equilibrium and mechanism-design properties.

Our main contributions are as follows.

\begin{enumerate}
    \item We formalize competition for conscious access as a contest among
    internal modules competing for a scarce broadcast slot. This provides an explicit allocation model for a mechanism that is often described
    qualitatively in workspace and attention theories.

    \item We prove existence of a pure-strategy Nash equilibrium under standard convexity and bounded-benefit assumptions, and give sufficient conditions for uniqueness using diagonal strict concavity. These results identify how stability depends on competition intensity, benefit scale, number of modules, and effort-cost curvature.

    \item We characterize when a lower-value representation can capture access because it has a sufficient amplification advantage. In the two-module setting, we first derive a sharp threshold in the quadratic-cost case, and then prove a general if-and-only-if capture criterion for strongly convex costs under a natural monotonicity condition.

    \item We prove that, under the same curvature-dominance condition that guarantees uniqueness, the unique pure Nash equilibrium of the general
    \(M\)-module access contest can be approximated efficiently by projected pseudo-gradient dynamics, with logarithmic dependence on the desired
    accuracy.

    \item We prove an impossibility theorem showing that exact winner-take-all efficiency and robustness are incompatible for single-slot mechanisms. Consequently, no anonymous differentiable access rule can simultaneously satisfy value monotonicity, robustness, and exact efficiency consistency. This provides a structural reason for studying smooth probabilistic access rules.
\end{enumerate}

\section{The Access Contest}
\label{sec:model}

We model conscious access as a single-slot allocation problem. A set of
candidate representations is produced by internal modules, but only one
representation can gain access to the global broadcast channel at a given
decision point. Modules may invest costly amplification effort to increase their
probability of access.

\subsection{Modules, scores, and access}

There are \(M\ge 2\) modules, indexed by \(i\in\{1,\dots,M\}\). Module \(i\)
has intrinsic value
$v_i\in\mathbb{R}$
which represents the task relevance or downstream value of its candidate
representation if it becomes consciously accessible.
Each module chooses an amplification effort $e_i\in[0,\infty).$

Effort may be interpreted as attention, precision, recurrent amplification, or
working-memory gating. The effective score of module \(i\) is $s_i=v_i+e_i.$

Access is allocated according to the logit contest success function:
\[
\sigma_i(e)
=
\frac{\exp(\beta s_i)}
{\sum_{j=1}^M \exp(\beta s_j)},
\; \beta>0,
\]
where \(e=(e_1,\dots,e_M)\). Thus \(\sigma_i(e)\) is the probability that
module \(i\) gains access.

The parameter \(\beta\) controls competition intensity. As \(\beta\to 0\), the
access probabilities approach the uniform distribution. As \(\beta\to\infty\),
the rule approaches winner-take-all selection among the highest-scoring
modules. Thus the model interpolates between diffuse access and sharp
competition for a single broadcast slot.

The logit access rule is used as a minimal smooth formalisation of competitive selection in global workspace and biased-competition accounts. Global workspace
theories describe conscious access as the amplification and global broadcast of selected representations, while biased-competition models describe attention as competition among representations under limited processing capacity
\cite{baars1988cognitive,dehaene2001cognitive,mashour2020gnw,desimone1995attention}.

\subsection{Costs and payoffs}

Effort is costly. Each module \(i\) has a cost function
$c_i:[0,\infty)\to\mathbb{R},$
which is assumed to be increasing and convex. If module \(i\) gains access, it
receives benefit
$B_i:=B_i(v_i)\ge 0.$
The expected payoff of module \(i\) is therefore
$u_i(e)=\sigma_i(e)B_i-c_i(e_i).$

The resulting model is a continuous-action normal-form game. It is not a finite
game, since each module chooses effort from the continuum \([0,\infty)\).
Formally,
$G=
\left(
\{1,\dots,M\},
\{[0,\infty)\}_{i=1}^M,
\{u_i\}_{i=1}^M
\right).$

A pure-strategy Nash equilibrium is an effort profile
\(e^\star\in[0,\infty)^M\) such that, for every module \(i\) and every
alternative effort \(e_i\ge 0\),
$u_i(e_i^\star,e_{-i}^\star)
\ge
u_i(e_i,e_{-i}^\star).$





The primitives of the model have direct interpretations. The values \(v_i\)
encode which candidate representations are useful in the current context. The
costs \(c_i\) encode how difficult each representation is to amplify or maintain.
The parameter \(\beta\) controls how sharply the access rule responds to score
differences.

This creates two possible sources of distortion. First, if competition is very
sharp, small score differences may have large effects on access probabilities.
Second, a lower-value representation may still gain access if it is sufficiently
cheap to amplify. The capture results in  Section~\ref{sec:two_module_capture} formalise
this second effect.
\section{Equilibrium Guarantees}
\label{sec:equilibrium}

Before studying capture and computation, we first establish that the access contest is well posed. In particular, we show that an equilibrium exists under
standard convexity assumptions, and we give interpretable sufficient conditions for uniqueness.

Throughout this section, fix a context and write
$B_i:=B_i(v_i).$
We impose the following assumptions.

\begin{assumption}[Costs]
\label{ass:costs}
For each \(i\), the cost function \(c_i:[0,\infty)\to\mathbb{R}\) is twice
continuously differentiable, nondecreasing, and strongly convex. That is, there
exists \(m_i>0\) such that
$c_i''(e)\ge m_i
\;\text{for all }e\ge 0.$
Moreover,
$c_i'(e)\to\infty
\;\text{as }e\to\infty.$
\end{assumption}

\begin{assumption}[Bounded benefits]
\label{ass:benefits}
There exists \(\bar B<\infty\) such that
\[
0\le B_i\le \bar B
\;\text{for all }i\in\{1,\dots,M\}.
\]
\end{assumption}

The following elementary bounds for the logit access rule are used throughout
the analysis.

\begin{lemma}[Derivative bounds for the logit access rule]
\label{lem:logit_bounds}
For every effort profile \(e\) and every module \(i\),
$\frac{\partial \sigma_i(e)}{\partial e_i}
=
\beta \sigma_i(e)\bigl(1-\sigma_i(e)\bigr),$
and hence
$0<
\frac{\partial \sigma_i(e)}{\partial e_i}
\le
\frac{\beta}{4}.$
Moreover,
$\frac{\partial^2 \sigma_i(e)}{\partial e_i^2}
=
\beta^2 \sigma_i(e)\bigl(1-\sigma_i(e)\bigr)
\bigl(1-2\sigma_i(e)\bigr),$
and
$\left|
\frac{\partial^2 \sigma_i(e)}{\partial e_i^2}
\right|
\le
\frac{\beta^2}{6\sqrt{3}}.$

\end{lemma}

\begin{proof}
The identities follow by direct differentiation of the logit access rule. The
bound
$\sigma_i(e)(1-\sigma_i(e))\le \frac14$
gives the first inequality. For the curvature bound, one maximises
$|p(1-p)(1-2p)|$
over \(p\in[0,1]\), whose maximum is \(1/(6\sqrt{3})\).
\end{proof}

The marginal gain from effort is therefore uniformly bounded. Since marginal
costs eventually become arbitrarily large, best responses are bounded.

\begin{lemma}[Bounded best responses]
\label{lem:bounded_best_responses}
Under Assumptions \ref{ass:costs} and \ref{ass:benefits}, for each module
\(i\) there exists \(\bar e_i<\infty\) such that every best response of module
\(i\) lies in \([0,\bar e_i]\). One admissible choice is any \(\bar e_i\)
satisfying
$c_i'(\bar e_i)\ge \frac{\beta \bar B}{4}.$
\end{lemma}

\begin{proof}
Fix \(e_{-i}\). For module \(i\),
$\frac{\partial u_i(e_i,e_{-i})}{\partial e_i}
=
B_i\frac{\partial \sigma_i(e)}{\partial e_i}
-
c_i'(e_i).$
By Lemma~\ref{lem:logit_bounds},
$B_i\frac{\partial \sigma_i(e)}{\partial e_i}
\le
\frac{\beta B_i}{4}
\le
\frac{\beta \bar B}{4}.$
Hence, whenever \(e_i>\bar e_i\),
$\frac{\partial u_i(e_i,e_{-i})}{\partial e_i}<0$.
Thus no best response can lie above \(\bar e_i\).
\end{proof}

Define the compact strategy set
$E:=\prod_{i=1}^M [0,\bar e_i],$
where each \(\bar e_i\) is chosen as in Lemma~\ref{lem:bounded_best_responses}.

\begin{theorem}[Existence of equilibrium]
\label{thm:existence}
Suppose Assumptions \ref{ass:costs} and \ref{ass:benefits} hold. Assume also
that, for every \(i\),
$m_i>
\frac{\beta^2 B_i}{6\sqrt{3}}.$
Then the access contest admits a pure-strategy Nash equilibrium
$e^\star\in E.$
\end{theorem}

\begin{proof}
By Lemma~\ref{lem:bounded_best_responses}, all best responses lie in the
compact convex set \(E\). For fixed \(e_{-i}\), the payoff of player \(i\) as a
function of \(e_i\) satisfies
\[
\frac{\partial^2 u_i(e_i,e_{-i})}{\partial e_i^2}
=
B_i\frac{\partial^2\sigma_i(e)}{\partial e_i^2}
-
c_i''(e_i).
\]
Using Lemma~\ref{lem:logit_bounds} and the assumed lower bound on \(c_i''\),
we obtain
\[
\frac{\partial^2 u_i(e_i,e_{-i})}{\partial e_i^2}
\le
\frac{\beta^2 B_i}{6\sqrt{3}}-m_i
<0.
\]
Thus each player's payoff is strictly concave in its own effort. Therefore each
best response is single-valued. Since payoffs are continuous and the feasible
sets are compact, Berge's maximum theorem implies that the best-response map
is continuous. Brouwer's fixed-point theorem then gives a fixed point of the
best-response map in \(E\), which is a pure-strategy Nash equilibrium.
\end{proof}

Existence ensures that the model has stable predictions. We next give a
standard uniqueness condition based on Rosen's diagonal strict concavity\cite{rosen1965concave}.

\begin{definition}[Pseudo-gradient]
\label{def:pseudogradient}
The pseudo-gradient of the access contest is
\[
g(e):=
\left(
\frac{\partial u_1(e)}{\partial e_1},
\dots,
\frac{\partial u_M(e)}{\partial e_M}
\right).
\]
\end{definition}

\begin{theorem}[Uniqueness via diagonal strict concavity]
\label{thm:rosen_uniqueness}
Assume the utilities are concave in own effort on \(E\). If the pseudo-gradient
\(g\) satisfies Rosen's diagonal strict concavity condition on \(E\), then the
Nash equilibrium in \(E\) is unique.
\end{theorem}

\begin{proof}
This is an application of Rosen's uniqueness theorem for concave \(N\)-person
games. Under diagonal strict concavity, two distinct equilibria would violate
the strict monotonicity of the weighted pseudo-gradient. Hence at most one
equilibrium can exist in \(E\). Together with Theorem~\ref{thm:existence}, this
gives uniqueness.
\end{proof}

For completeness, we state a simple sufficient condition that implies diagonal
strict concavity in the present model.

\begin{proposition}[A simple sufficient condition for uniqueness]
\label{prop:simple_uniqueness}
Let
$m:=\min_i \inf_{e\ge 0} c_i''(e),
\;
\bar B:=\max_i B_i.$
If
\[
m>
\beta^2 \bar B
\left(
\frac{M-1}{4}
+
\frac{1}{6\sqrt{3}}
\right),
\]
then the Nash equilibrium in \(E\) is unique.
\end{proposition}

\begin{proof}
Let \(J(e)=Dg(e)\) be the Jacobian of the pseudo-gradient and let
\[
S(e):=\frac{J(e)+J(e)^\top}{2}
\]
be its symmetric part. It is enough to show that \(S(e)\) is negative definite
for every \(e\in E\).

For the diagonal entries,
\[
J_{ii}(e)
=
B_i\frac{\partial^2\sigma_i(e)}{\partial e_i^2}
-
c_i''(e_i).
\]
By Lemma~\ref{lem:logit_bounds},
\[
J_{ii}(e)
\le
\frac{\beta^2\bar B}{6\sqrt{3}}-m.
\]
For \(i\neq j\), the mixed derivatives of the logit rule satisfy
\[
\left|
\frac{\partial^2\sigma_i(e)}{\partial e_i\partial e_j}
\right|
\le
\frac{\beta^2}{4},
\]
and therefore
\[
|J_{ij}(e)|\le \frac{\beta^2\bar B}{4}.
\]
The same bound applies to the off-diagonal entries of \(S(e)\). Hence, for
each row,
\[
S_{ii}(e)+\sum_{j\neq i}|S_{ij}(e)|
\le
\left(
\frac{\beta^2\bar B}{6\sqrt{3}}-m
\right)
+
(M-1)\frac{\beta^2\bar B}{4}.
\]
The assumed condition makes the right-hand side strictly negative. By
Gershgorin's circle theorem, all eigenvalues of \(S(e)\) are strictly negative,
so \(S(e)\) is negative definite. Therefore the pseudo-gradient is strictly
monotone decreasing on \(E\), which implies Rosen's diagonal strict concavity
condition and hence uniqueness.
\end{proof}

\begin{remark}[Interpretation]
Proposition~\ref{prop:simple_uniqueness} gives a transparent stability
condition. Larger competition intensity \(\beta\), larger benefit scale
\(\bar B\), or more modules \(M\) increase strategic coupling between players.
In contrast, stronger effort-cost curvature \(m\) stabilises the game. Thus the
equilibrium is unique when marginal costs grow sufficiently quickly relative to
the strength of competition for access.
\end{remark}

\subsection{Efficient computation of the pure Nash equilibrium}
\label{subsec:efficient_pne}
We now show that the same curvature-dominance condition that yields uniqueness also gives an efficient method for computing the equilibrium. The result applies to the general \(M\)-module access contest.
Recall the pseudo-gradient
g(e)=
$\left(
\frac{\partial u_1(e)}{\partial e_1},
\dots,
\frac{\partial u_M(e)}{\partial e_M}
\right),$
and define
$F(e):=-g(e).$
Let
$E=\prod_{i=1}^M[0,\bar e_i]$
be the compact box containing all best responses. Under the hypotheses of Proposition~\ref{prop:simple_uniqueness}, define the curvature margin
\[
\alpha
:=
m-
\beta^2\bar B
\left(
\frac{M-1}{4}
+
\frac{1}{6\sqrt{3}}
\right)
>0.
\]
Let \(L\) be any Lipschitz constant of \(F\) on \(E\), and let
$D:=\sup_{e,\tilde e\in E}\|e-\tilde e\|$
be the diameter of \(E\). Consider the projected pseudo-gradient iteration
$e^{t+1}
=
\Pi_E\bigl(e^t-\eta F(e^t)\bigr)
=
\Pi_E\bigl(e^t+\eta g(e^t)\bigr),
\;
\eta=\frac{\alpha}{L^2}.$

Such an \(L\) exists because \(F\) is continuously differentiable and \(E\) is compact.

\begin{theorem}[Efficient computation of the unique pure Nash equilibrium]
\label{thm:efficient_pne_computation}
Under the hypotheses of Proposition~\ref{prop:simple_uniqueness}, the access
contest has a unique pure Nash equilibrium \(e^\star\in E\). Moreover, the
projected pseudo-gradient iteration satisfies
$\|e^T-e^\star\|\le \varepsilon$
after
$T
=
O\!\left(
\frac{L^2}{\alpha^2}
\log\frac{D}{\varepsilon}
\right)$
iterations. If each \(c_i'\) can be evaluated in constant time, the total
runtime complexity is
$O\!\left(
M\frac{L^2}{\alpha^2}
\log\frac{D}{\varepsilon}
\right).$
\end{theorem}

\begin{proof}
By Proposition~\ref{prop:simple_uniqueness}, the symmetric part of the
Jacobian of the pseudo-gradient is uniformly negative definite on \(E\). In
particular,
\[
\frac{Dg(e)+Dg(e)^\top}{2}
\preceq
-\alpha I
\;
\text{for all }e\in E.
\]
Since \(F=-g\), this implies
\[
\frac{DF(e)+DF(e)^\top}{2}
\succeq
\alpha I,
\]
and therefore \(F\) is strongly monotone:
$(F(e)-F(\tilde e))^\top(e-\tilde e)
\ge
\alpha\|e-\tilde e\|^2
\;
\text{for all }e,\tilde e\in E.$
Because each payoff is concave in the player's own effort on \(E\), a profile
\(e^\star\in E\) is a Nash equilibrium if and only if it solves the variational
inequality
$(e-e^\star)^\top F(e^\star)\ge 0
\;
\text{for all }e\in E.$
Equivalently, \(e^\star\) is a fixed point of the projected map
$T(e):=\Pi_E(e-\eta F(e)).$
We now show that \(T\) is a contraction. Since projection onto a closed convex
set is nonexpansive,

$\begin{aligned}
\|T(e)-T(\tilde e)\|^2
&\le
\|(e-\eta F(e))-(\tilde e-\eta F(\tilde e))\|^2 \\
&=
\|e-\tilde e\|^2
-2\eta(F(e)-F(\tilde e))^\top(e-\tilde e)
+\eta^2\|F(e)-F(\tilde e)\|^2 .
\end{aligned}$

Using strong monotonicity and \(L\)-Lipschitz continuity,
$\|T(e)-T(\tilde e)\|^2
\le
\left(1-2\eta\alpha+\eta^2L^2\right)\|e-\tilde e\|^2.$

With \(\eta=\alpha/L^2\),
$\|T(e)-T(\tilde e)\|^2
\le
\left(1-\frac{\alpha^2}{L^2}\right)\|e-\tilde e\|^2.$
Thus \(T\) is a contraction. Since its unique fixed point is \(e^\star\),
\[
\|e^t-e^\star\|
\le
\left(1-\frac{\alpha^2}{L^2}\right)^{t/2}
\|e^0-e^\star\|
\le
\left(1-\frac{\alpha^2}{L^2}\right)^{t/2}D.
\]
Using \(1-x\le e^{-x}\), it is enough to take
\[
t
=
O\!\left(
\frac{L^2}{\alpha^2}
\log\left(\frac{D}{\varepsilon}\right)
\right)
\]
to ensure \(\|e^t-e^\star\|\le\varepsilon\).

Finally, each iteration computes the logit probabilities, the pseudo-gradient,
and the projection onto \(E\). These operations take \(O(M)\) time if
each \(c_i'\) is evaluated in constant time. Multiplying by the number of iterations gives the stated total cost.
\end{proof}

\section{Capture in the Two-Module Case}
\label{sec:two_module_capture}

We now specialise in the two-module case. This setting is simple enough to provide sharp characterisations, but still captures the key phenomenon of interest, that is a lower-value representation may gain access because it is easier to amplify.

Let \(M=2\), and suppose
$v_1>v_2$.

Thus module \(1\) has the intrinsically higher-value representation, while
module \(2\) is lower-value. Define the intrinsic value gap
$\delta:=v_1-v_2>0.$

Write
$\Delta:=(v_1+e_1)-(v_2+e_2)$
for the score gap between the two modules. Since
\[
\sigma_1(e)=\frac{1}{1+\exp(-\beta\Delta)},
\;
\sigma_2(e)=1-\sigma_1(e),
\]
module \(2\) receives more than half of the access probability if and only if
$\sigma_2(e)>\frac12
\;\Longleftrightarrow\;
\sigma_1(e)<\frac12
\;\Longleftrightarrow\;
\Delta<0.$
We use the following definition:

\begin{definition}[Capturing interior equilibrium]
\label{def:capturing_equilibrium}
In the two-module setting with \(v_1>v_2\), an interior Nash equilibrium
\(e^\star=(e_1^\star,e_2^\star)\) is called a \emph{capturing interior
equilibrium} if the lower-value module obtains more than half of the access
probability: $\sigma_2(e^\star)>\frac12.$
Equivalently, since
$\sigma_1(e)=\frac{1}{1+\exp(-\beta\Delta)},
and \;
\Delta=(v_1+e_1)-(v_2+e_2),$
capture occurs if and only if
$\Delta^\star<0.$
\end{definition}

Throughout this section, assume that for each \(i\in\{1,2\}\), the cost
function \(c_i\) is \(C^2\), strictly convex, and satisfies
\[
c_i(0)=0,\; c_i'(0)=0,\; c_i''(e)\ge m_i>0
\;\text{for all }e\ge 0.
\]
We also assume that each player's payoff is strictly concave in its own effort,
for example
\[
c_i''(e)>\frac{\beta^2 B_i}{6\sqrt{3}}
\;
\text{for all }e\ge 0,\ i\in\{1,2\}.
\]
This ensures that the interior first-order conditions are sufficient for
best-response optimality.

\subsection{A general if-and-only-if capture criterion}

The first result gives an exact capture criterion for general strongly convex
costs. It is useful to define the inverse marginal-cost response of each module.

\begin{definition}[Inverse marginal-cost responses and amplification advantage]
\label{def:amplification_advantage}
For \(q\in[0,1/4]\), define
$\psi_i(q):=(c_i')^{-1}(\beta B_i q),
\; i\in\{1,2\},$
and define the amplification-advantage function
$A(q):=\psi_2(q)-\psi_1(q).$
\end{definition}

The quantity \(A(q)\) measures how much more effort module \(2\) supplies than
module \(1\) when both face the same logit sensitivity term \(q\). Thus \(A(q)\)
captures the cost-adjusted amplification advantage of the lower-value module.

\begin{theorem}[General if-and-only-if capture criterion]
\label{thm:general_iff_capture}
Suppose the function \(A\) from Definition~\ref{def:amplification_advantage}
is nondecreasing on \([0,1/4]\). Then a capturing interior equilibrium exists
if and only if
$A(1/4)>\delta.$
At the boundary case \(A(1/4)=\delta\), there exists an interior equilibrium
with balanced access,
$\sigma_1=\sigma_2=\frac12.$
If \(A(1/4)<\delta\), no capturing interior equilibrium exists. Moreover,
whenever \(A(1/4)>\delta\), the capturing interior equilibrium is unique.
\end{theorem}

\begin{proof}
Write
\[
p:=\sigma_1(e)
=
\frac{1}{1+\exp(-\beta\Delta)},
\;
\sigma_2(e)=1-p,
\]
and define $q:=p(1-p)$.
As noted above, capture by module \(2\) is equivalent to
$p<\frac12
\;\Longleftrightarrow\;
\Delta<0.$
We first characterise interior equilibria. The payoffs are
$u_1(e_1,e_2)=pB_1-c_1(e_1),
\;
u_2(e_1,e_2)=(1-p)B_2-c_2(e_2).$
For the two-module logit rule,
$\frac{\partial p}{\partial e_1}
=
\beta p(1-p)
=
\beta q,$
and
$\frac{\partial (1-p)}{\partial e_2}
=
\beta p(1-p)
=
\beta q.$

Hence the interior first-order conditions are
$c_1'(e_1)=\beta B_1q,
\;
c_2'(e_2)=\beta B_2q.$
Because each \(c_i\) is strictly convex, \(c_i'\) is strictly increasing and
therefore invertible on its range. Thus
$e_1=\psi_1(q),
\;
e_2=\psi_2(q),$
and consequently
$e_2-e_1=A(q).$

Now the score gap satisfies
\[
\Delta
=
(v_1+e_1)-(v_2+e_2)
=
(v_1-v_2)+(e_1-e_2)
=
\delta-A(q).
\]
Since
\[
p=\frac{1}{1+\exp(-\beta\Delta)},
\]
we have the identity
\[
q=p(1-p)
=
\frac{1}{4\cosh^2(\beta\Delta/2)}.
\]
Therefore every interior equilibrium corresponds to a root of the scalar
equation
$H(\Delta)=0,$
where
\[
H(\Delta)
:=
\Delta-\delta+
A\!\left(
\frac{1}{4\cosh^2(\beta\Delta/2)}
\right).
\]

We now study \(H\) on the capture region \((-\infty,0]\). Define
\[
q(\Delta)
:=
\frac{1}{4\cosh^2(\beta\Delta/2)}.
\]
For \(\Delta\le 0\), the function \(q(\Delta)\) is nondecreasing. Indeed,
$q'(\Delta)
=
-\beta q(\Delta)\tanh(\beta\Delta/2)$,
and \(\tanh(\beta\Delta/2)\le 0\) whenever \(\Delta\le 0\), so
$q'(\Delta)\ge 0.$

Because \(A\) is nondecreasing, the composition \(A(q(\Delta))\) is also
nondecreasing on \((-\infty,0]\). Therefore \(H\) is strictly increasing on
\((-\infty,0]\). More explicitly, if \(\Delta_1<\Delta_2\le 0\), then
$q(\Delta_1)\le q(\Delta_2)$,
so
$A(q(\Delta_1))\le A(q(\Delta_2))$,
and hence
$H(\Delta_2)-H(\Delta_1)
=
(\Delta_2-\Delta_1)
+
\bigl(A(q(\Delta_2))-A(q(\Delta_1))\bigr)
>0.$

Next,
$\lim_{\Delta\to-\infty}q(\Delta)=0.$
Since \(c_i'(0)=0\), we have
$\psi_i(0)=0,$
and hence
$A(0)=0.$
Thus
$\lim_{\Delta\to-\infty}H(\Delta)
=
-\infty.$
At the boundary point \(\Delta=0\),
$H(0)
=
-\delta+A(1/4).$

We now distinguish three cases:
\begin{itemize}
    \item If \(A(1/4)>\delta\), then \(H(0)>0\). Since \(H(\Delta)\to-\infty\) as
    \(\Delta\to-\infty\), and since \(H\) is continuous and strictly increasing on
    \((-\infty,0]\), there exists a unique root
    $\Delta^\star<0.$
    This root corresponds to a unique capturing interior equilibrium.
    \item If \(A(1/4)=\delta\), then \(H(0)=0\), so \(\Delta=0\) is a root. Since
    \(\Delta=0\) implies \(p=1/2\), the corresponding equilibrium satisfies
    $\sigma_1=\sigma_2=\frac12.$
    \item If \(A(1/4)<\delta\), then \(H(0)<0\). Since \(H\) is strictly increasing and
    tends to \(-\infty\) as \(\Delta\to-\infty\), it follows that
    $H(\Delta)<0
    \;
    \text{for all }\Delta\le 0.$
\end{itemize}

Thus, there is no root in the capture region, and hence no capturing interior
equilibrium.

It remains only to check that a root of \(H\) indeed defines an equilibrium.
Let \(\Delta^\star\) be a root of \(H\), and set
\[
q^\star:=
\frac{1}{4\cosh^2(\beta\Delta^\star/2)}.
\]
Define
$e_i^\star:=\psi_i(q^\star),
\; i\in\{1,2\}.$
Then
$c_i'(e_i^\star)=\beta B_iq^\star,
\; i\in\{1,2\},$
so the first-order conditions hold. Moreover, since \(H(\Delta^\star)=0\),
$\Delta^\star
=
\delta-A(q^\star)
=
(v_1-v_2)+(e_1^\star-e_2^\star),$
so the score-gap consistency condition also holds. Therefore
\((e_1^\star,e_2^\star)\) is a stationary point of the game. By the standing
own-payoff strict concavity assumption, these first-order conditions are
sufficient for best-response optimality. Hence \((e_1^\star,e_2^\star)\) is an
interior Nash equilibrium.

This proves all claims.
\end{proof}

\begin{remark}[Interpretation]
Theorem~\ref{thm:general_iff_capture} says that capture occurs exactly when the maximum amplification advantage of the lower-value module exceeds the
intrinsic value gap. The term \(A(1/4)\) is the largest possible effort advantage that module \(2\) can obtain from the equilibrium first-order conditions, because \(q=p(1-p)\le 1/4\). Thus the condition
$A(1/4)>\delta$
has a direct interpretation. Lower-value content captures access precisely when its cost-adjusted amplification advantage is large enough to overcome its
initial value disadvantage.
\end{remark}

\subsection{Quadratic costs and the sharp competition threshold}

The general criterion recovers a simple closed-form threshold under quadratic
costs. This case is useful because it makes the dependence on the competition
intensity \(\beta\) explicit.

Assume
$c_i(e)=\frac{\gamma_i}{2}e^2,
\;
\gamma_i>0,$
and consider content-valued benefits 
$B_i=v_i.$
Define the cost-adjusted amplification asymmetry
$a:=
\frac{v_2}{\gamma_2}
-
\frac{v_1}{\gamma_1}.$
The quantity \(a\) measures whether the lower-value module has a sufficient
cost advantage. If \(a\le 0\), then module \(2\) does not have a cost-adjusted
amplification advantage over module \(1\). If \(a>0\), module \(2\) can
potentially compensate for its lower intrinsic value by amplifying more cheaply.

\begin{corollary}[Sharp threshold under quadratic costs]
\label{cor:quadratic_capture_threshold}
Suppose \(c_i(e)=\frac{\gamma_i}{2}e^2\), \(B_i=v_i\), and the standing
own-payoff concavity assumptions hold. If \(a\le 0\), no capturing interior
equilibrium exists. If \(a>0\), define
$\beta^\star
:=
\frac{4(v_1-v_2)}
{v_2/\gamma_2-v_1/\gamma_1}.$
Then:
$
\beta<\beta^\star
\;\Longrightarrow\;
\text{no capturing interior equilibrium exists,}$
$\beta=\beta^\star
\;\Longrightarrow\;
\sigma_1=\sigma_2=\frac12,$
and
$\beta>\beta^\star
\;\Longrightarrow\;
\text{there exists a unique capturing interior equilibrium.}$

\end{corollary}

\begin{proof}
For quadratic costs,
$c_i'(e)=\gamma_i e,$
so
$(c_i')^{-1}(x)=\frac{x}{\gamma_i}$.
Therefore
$\psi_i(q)
=
(c_i')^{-1}(\beta v_iq)
=
\frac{\beta v_iq}{\gamma_i}.$
Hence
$A(q)
=
\psi_2(q)-\psi_1(q)
=
\beta q
\left(
\frac{v_2}{\gamma_2}
-
\frac{v_1}{\gamma_1}
\right)
=
\beta aq.$
Thus
$A(1/4)=\frac{\beta a}{4}.$

If \(a\le 0\), then \(A(1/4)\le 0<\delta\), so Theorem~\ref{thm:general_iff_capture}
implies that no capturing interior equilibrium exists.

If \(a>0\), then Theorem~\ref{thm:general_iff_capture} gives capture if and
only if
$A(1/4)>\delta.$
Substituting \(A(1/4)=\beta a/4\), this becomes
$\frac{\beta a}{4}>\delta,$
or equivalently
$\beta>\frac{4\delta}{a}.$
Since \(\delta=v_1-v_2\), this is exactly
\[
\beta>
\frac{4(v_1-v_2)}
{v_2/\gamma_2-v_1/\gamma_1}
=
\beta^\star.
\]
The boundary case \(\beta=\beta^\star\) gives \(A(1/4)=\delta\), and therefore
balanced access \(\sigma_1=\sigma_2=1/2\). The result follows.
\end{proof}

\begin{remark}[Phase transition]
Corollary~\ref{cor:quadratic_capture_threshold} gives a phase-transition
interpretation. Below the threshold \(\beta^\star\), competition is not sharp
enough for the lower-value module's amplification advantage to overturn the
intrinsic value gap. At the threshold, access is balanced. Above the threshold,
the lower-value module captures access at the unique capturing interior
equilibrium.
\end{remark}

\section{Limits of Single-Slot Access Mechanisms}
\label{sec:impossibility}

The previous sections analysed a specific smooth access rule, namely the logit
contest success function. We now step back and ask a more general question:
what properties can any single-slot access mechanism satisfy?

This section shows that there is a fundamental tension between exact
winner-take-all efficiency and robustness. If a mechanism always assigns
probability one to the unique highest-scoring module, then it must be
discontinuous at ties. Thus, any robust and differentiable access rule must
smooth the transition between competing modules. This provides a structural
justification for studying probabilistic access mechanisms such as the logit
rule.

\begin{definition}[Single-slot access mechanism]
\label{def:single_slot_mechanism}
A single-slot access mechanism is a map
$p:\mathbb{R}^M\to \Delta^{M-1},
\;
s\mapsto p(s)=(p_1(s),\dots,p_M(s)),$
where \(s=(s_1,\dots,s_M)\) is the vector of module scores and \(p_i(s)\) is
the probability that module \(i\) receives the unique access slot.
\end{definition}

\begin{definition}[Exact efficiency consistency]
\label{def:exact_efficiency_consistency}
A single-slot access mechanism \(p\) is \emph{exactly efficiency consistent} if
whenever module \(i\) is the unique score maximiser, that module receives the
access slot with probability one:
$s_i>s_j \ \text{for all } j\neq i
\;\Longrightarrow\;
p_i(s)=1.$
\end{definition}

\begin{definition}[Budget robustness]
\label{def:budget_robustness}
A single-slot access mechanism \(p\) is \emph{budget robust} if it is
continuous as a function of the score vector \(s\). In other words, small
perturbations in scores induce only small perturbations in access
probabilities.
\end{definition}

\begin{definition}[Value monotonicity]
\label{def:value_monotonicity}
A single-slot access mechanism \(p\) satisfies \emph{value monotonicity} if,
for every \(i\), increasing only module \(i\)'s score cannot decrease its
access probability. Formally, if \(s,\tilde s\in\mathbb{R}^M\) satisfy
$\tilde s_i>s_i,
\;
\tilde s_j=s_j \ \text{for all } j\neq i,$
then
$p_i(\tilde s)\ge p_i(s).$

\end{definition}

\begin{definition}[Anonymity]
\label{def:anonymity}
A single-slot access mechanism \(p\) is \emph{anonymous} if relabelling the
modules only relabels the resulting access probabilities. Equivalently, for
every permutation \(\pi\) of \(\{1,\dots,M\}\), \;
$p_{\pi(i)}(\pi s)=p_i(s)
\;\text{for all }i,$
where
$(\pi s)_j:=s_{\pi^{-1}(j)}.$

\end{definition}

\begin{theorem}[Impossibility of exact efficiency and robustness]
\label{thm:impossibility_single_slot}
No budget-robust single-slot access mechanism is exactly efficiency consistent.

Consequently, no anonymous differentiable single-slot access mechanism can
simultaneously satisfy value monotonicity, budget robustness, and exact
efficiency consistency.
\end{theorem}

\begin{proof}
Assume, for contradiction, that there exists a single-slot access mechanism
$p:\mathbb{R}^M\to\Delta^{M-1}$
that is both budget robust and exactly efficiency consistent.

Fix two distinct modules, say modules \(1\) and \(2\). Let \(K>0\), and
consider the one-parameter family of score vectors
$s(x):=(x,0,-K,\dots,-K)\in\mathbb{R}^M,
\; x\in\mathbb{R}$.
Define
$f(x):=p_1(s(x)).$
We analyse \(f\) near \(x=0\).
If \(x>0\), then module \(1\) is the unique score maximiser in \(s(x)\), since
$x>0>-K.$
By exact efficiency consistency,
$f(x)=p_1(s(x))=1
\;\text{for all }x>0.$

If \(-K<x<0\), then module \(2\) is the unique score maximiser in \(s(x)\),
since
$0>x>-K.$

Again by exact efficiency consistency,
$p_2(s(x))=1
\;\text{for all }x\in(-K,0).$

Since the coordinates of \(p(s(x))\) sum to one, this implies
$f(x)=p_1(s(x))=0
\;\text{for all }x\in(-K,0).$
Therefore,
$f(x)=1 \;\text{for all }x>0,
\;
f(x)=0 \;\text{for all }x\in(-K,0).$
Hence
$\lim_{x\downarrow 0} f(x)=1,
\;
\lim_{x\uparrow 0} f(x)=0.$
Thus \(f\) is discontinuous at \(x=0\).

However, the map \(x\mapsto s(x)\) is continuous, and \(p\) is budget robust,
hence continuous. Therefore the composition
$f(x)=p_1(s(x))$
must also be continuous. This is a contradiction.
Thus, no budget-robust single-slot access mechanism can be exactly efficiency consistent.

For the second claim, differentiability implies continuity. Therefore, no
differentiable single-slot access mechanism can satisfy both budget robustness and exact efficiency consistency. Adding anonymity and value monotonicity does not remove this contradiction. Hence, no anonymous differentiable single-slot access mechanism can simultaneously satisfy value monotonicity, budget robustness, and exact efficiency consistency.
\end{proof}

\begin{remark}[Interpretation]
Theorem~\ref{thm:impossibility_single_slot} shows that exact winner-take-all selection and robustness are incompatible in the single-slot setting. A deterministic argmax rule is exactly efficiency consistent, but it is necessarily
discontinuous at score ties. Conversely, a smooth probabilistic rule can be
robust, but it cannot assign probability one to the unique maximiser at every score profile. Thus, the use of a smooth access rule is not merely analytically
convenient. It is structurally motivated by the need for stability under small perturbations of scores.
\end{remark}

\begin{remark}[Connection to the logit access rule]
The logit access rule used in this paper satisfies budget robustness and value
monotonicity, but it does not satisfy exact efficiency consistency for finite
\(\beta\). Instead, it approaches exact winner-take-all selection only in the
limit as \(\beta\to\infty\). The impossibility theorem explains why this is unavoidable. Exact efficiency and robustness cannot both hold in a single-slot mechanism.
\end{remark}
\section{Discussion}
\label{sec:discussion}

We have modelled conscious access as a contest for a scarce broadcast slot. The aim is not to reduce consciousness to a single game, but to isolate one formal component of conscious access: competition among candidate representations for
global availability. This turns a common qualitative claim in workspace and attention theories into an allocation problem with equilibrium, capture, computation, and impossibility results.

The capture results show that access need not track intrinsic value alone. In the two-module setting, a lower-value representation captures access when its amplification advantage is large enough to overcome the intrinsic value gap. The general criterion expresses this through the function
$A(q)=\psi_2(q)-\psi_1(q),$
which measures the cost-adjusted amplification advantage of the lower-value module. Capture occurs precisely when
$A(1/4)>v_1-v_2.$
Thus, salience-driven or attention-driven capture can be interpreted as an
equilibrium consequence of asymmetric amplification costs, rather than merely
as a failure of selection.

In the quadratic-cost case, the capture condition becomes a sharp threshold in the competition intensity \(\beta\). Below the threshold, access remains value-tracking; above it, the lower-value module can dominate access at the unique capturing equilibrium. This gives a phase-transition interpretation. Increasing competition intensity can improve selectivity, but can also make the system more vulnerable to cheap-to-amplify representations.

The impossibility theorem gives a complementary reason for using smooth
probabilistic access rules. A deterministic argmax rule is exactly
winner-take-all, but it is discontinuous at ties. Hence, no single-slot access mechanism can be both exactly winner-take-all efficient and robust to small score perturbations. The logit rule smooths this transition. For finite
\(\beta\), it remains robust but does not assign probability one to the highest score. In this sense, smoothing is not only analytically convenient, but structurally necessary if robustness is required.

Finally, the computation result shows that the model is algorithmically
tractable in the uniqueness regime. Under the curvature-dominance condition, the unique pure Nash equilibrium of the general \(M\)-module access contest can
be approximated efficiently by projected pseudo-gradient dynamics. The same parameters that govern stability also govern tractability. Stronger cost curvature improves conditioning, while larger competition intensity, larger benefit scale, or more modules increase strategic coupling.
\section{Conclusion}
\label{sec:conclusion}

This paper introduced a game-theoretic framework for modelling competition for conscious access. We formalised access as a contest among internal modules competing for a scarce broadcast slot, established equilibrium guarantees under interpretable conditions, characterised when lower-value content can capture access through amplification advantages, proved efficient computation of the unique pure Nash equilibrium under curvature-dominance assumptions, and showed that robust single-slot access rules cannot be exactly winner-take-all.
Together, these results show that competition for conscious access can be studied as a formal allocation problem. The framework does not replace cognitive or neural theories of access; rather, it complements them by making explicit the strategic and algorithmic structure of selection under a scarce broadcast constraint.
The model is intentionally stylised. We focus on a single broadcast slot, static one-shot interaction, and effort choices by modules with fixed values and costs. This abstraction is useful because it allows the access problem to be analysed with standard tools from game theory, but it also leaves several important features outside the present model.

\subsection{Limitations and Future Work}First, conscious access may involve more than one representation at a time. Extending the model to \(K>1\) simultaneous access slots would allow one to study richer forms of partial access, shared broadcast, or top-\(K\) selection. Second, the current model is static. In realistic cognitive systems, values, costs, and access probabilities may evolve over time through learning, adaptation, fatigue, or changing task demands. A dynamic version of the model could study repeated access contests and the long-run stability of capture. Third, we treat module values and costs as fixed primitives. Future work could study how these quantities are inferred, learned, or endogenously shaped by task context and prior experience.

A further direction is welfare and efficiency. The present paper focuses on equilibrium, capture, computation, and mechanism-level impossibility. A natural next step is to compare equilibrium access with a planner benchmark and derive welfare-loss or price-of-anarchy bounds in terms of the number of modules, competition intensity, benefit scale, and cost curvature. This would quantify the cost of internal competition more directly.

Finally, the logit access rule is only one possible smooth contest success function. It is natural because it provides a smooth approximation to winner-take-all selection and connects to probabilistic choice, but other access mechanisms may be appropriate in different cognitive or neural settings. Extending the analysis to broader classes of contest success functions would clarify which results are specific to the logit rule and which are structural properties of access competition more generally.

\bibliography{bib}

\end{document}